\pdfoutput=1
\documentclass[11pt]{article}
\usepackage{acl}
\usepackage{xcolor}
\usepackage{times}
\usepackage{latexsym}
\usepackage[T1]{fontenc}
\usepackage[utf8]{inputenc}
\usepackage{microtype}
\usepackage{graphicx}
\usepackage{enumerate}
\usepackage{acronym}
\usepackage{booktabs}
\usepackage{placeins}
\usepackage{multirow}
\usepackage{makecell}
\usepackage{subfigure}
\usepackage{enumitem}
\usepackage{amsmath}
\usepackage{bm}
\usepackage{color}
\usepackage{amsthm}
\usepackage{mathtools}
\usepackage{fixmath}
\usepackage{siunitx}
\usepackage{geometry}
\acrodef{IR}{information retrieval}
\acrodef{PLM}{pre-trained language model}
\acrodef{NQ}{Natural Questions}
\acrodef{ER}{Eliminative Retrieval}
\acrodef{ExcluIR}{Exclusion in Information Retrieval}
\acrodef{RAG}{Retrieval-Augmented Generation}
\acrodef{DR}{Dense Retrieval}
\acrodef{GR}{Generative Retrieval}
\acrodef{QuePR}{Query Product Retriever}
\acrodef{NevIR}{Negation in Neural Information Retrieval}
\acrodef{T2VR}{Text-to-Video Retrieval}
\acrodef{MRR}{Mean Reciprocal Rank}
\acrodef{RR}{Right Rank}
\author{
    Wenhao Zhang\textsuperscript{\rm 1}, \ 
    Mengqi Zhang\textsuperscript{\rm 1}, \ 
    Shiguang Wu\textsuperscript{\rm 1}, \ 
    Jiahuan Pei\textsuperscript{\rm 2}, \ 
    Zhaochun Ren\textsuperscript{\rm 3},\\\ 
    \textbf{Maarten de Rijke\textsuperscript{\rm 4},} \ 
    \textbf{Zhumin Chen\textsuperscript{\rm 1},} \ 
    \textbf{Pengjie Ren\textsuperscript{\rm 1}\thanks{$^*$ Corresponding authors.}} \\\ 
    \textsuperscript{\rm 1} Shandong University, Qingdao, China \\
    \textsuperscript{\rm 2} Centrum Wiskunde \& Informatica, Amsterdam, The Netherlands \\
    \textsuperscript{\rm 3} Leiden University, Leiden, The Netherlands \\
    \textsuperscript{\rm 4} University of Amsterdam, Amsterdam, The Netherlands. \\
    \texttt{\{zhangwenhao,shiguang.wu\}@mail.sdu.edu.cn}, \\
    \texttt{\{mengqi.zhang,chenzhumin,renpengjie\}@sdu.edu.cn},\\
    \texttt{jiahuan.pei@cwi.nl},
    \texttt{\ z.ren@liacs.leidenuniv.nl},
    \texttt{\ m.derijke@uva.nl}
}
\title{ExcluIR: Exclusionary Neural Information Retrieval}

\newtheorem{definition}{Definition}
\newtheorem{assumption}{Assumption}
\newtheorem{claim}{Claim}

\newcommand{\qab}{q_{A,B}}
\newcommand{\qba}{q_{B,A}}
\newcommand{\da}{d_{A}}
\newcommand{\db}{d_{B}}

\newcommand{\innerprod}[2]{\langle #1, #2 \rangle}

\begin{document}
\maketitle
\begin{abstract}
Exclusion is an important and universal linguistic skill that humans use to express what they do not want.
However, in information retrieval community, there is little research on exclusionary retrieval, where users express what they do not want in their queries.
In this work, we investigate the scenario of exclusionary retrieval in document retrieval for the first time.
We present \acs{ExcluIR}, a set of resources for exclusionary retrieval, consisting of an evaluation benchmark and a training set for helping retrieval models to comprehend exclusionary queries.
The evaluation benchmark includes 3,452 high-quality exclusionary queries, each of which has been manually annotated.
The training set contains 70,293 exclusionary queries, each paired with a positive document and a negative document.
We conduct detailed experiments and analyses, obtaining three main observations: 
(1) Existing retrieval models with different architectures struggle to effectively comprehend exclusionary queries; 
(2) Although integrating our training data can improve the performance of retrieval models on exclusionary retrieval, there still exists a gap compared to human performance;
(3) Generative retrieval models have a natural advantage in handling exclusionary queries.
To facilitate future research on exclusionary retrieval, we share the benchmark and evaluation scripts on \url{https://github.com/zwh-sdu/ExcluIR}.
\end{abstract}
\section{Introduction}
\acresetall
\begin{figure}[t]
    \centering
    \includegraphics[scale=0.4]{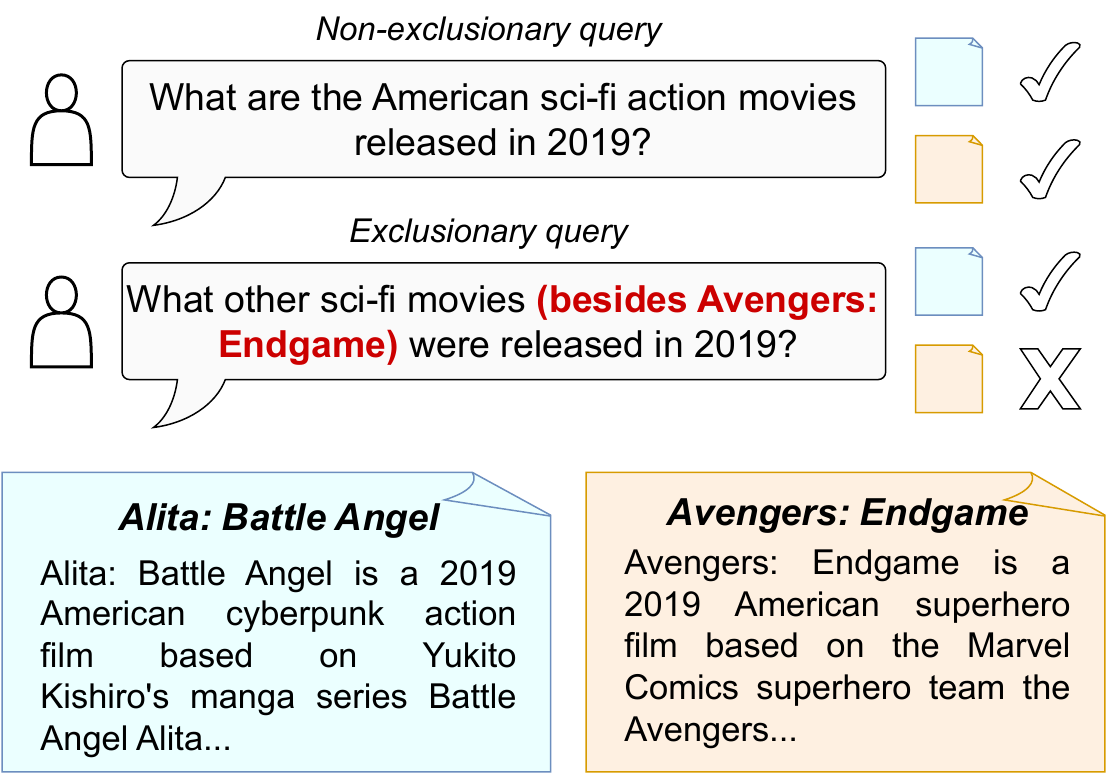}
    \caption{A comparison between non-exclusionary and exclusionary queries. 
    Exclusionary queries often specify content to be excluded (e.g., ``Avengers: Endgame'') to express the user's requirements for omitting certain information.
    In this case, if the retrieval system fails to comprehend the exclusionary nature of a query (e.g., one containing the term ``besides,'') it will produce retrieval results that users do not desire.
    }
    \label{fig1}
\vspace{-3mm}
\end{figure}
Selective attention~\cite{treisman1964selective, attention, cherry2020we}, defined as the ability to focus on relevant information while disregarding irrelevant information, plays a crucial role in shaping user's search behaviors.
This principle, originating from cognitive psychology, not only shapes human perception of the environment but also extends its influence to interactions with information retrieval systems.
When searching for information, users can express a desire to exclude certain information.
We refer to this phenomenon as \textit{exclusionary retrieval}, where users explicitly indicate their preference to exclude specific information.

Exclusionary retrieval emphasizes a crucial need for precision and relevance in information retrieval.
It shows how users leverage their knowledge and expectations to find information that meets their specific needs.
Therefore, the failure to understand exclusionary queries can present a potentially serious problem.
For example, as shown in Figure~\ref{fig1}, imagine a user searching for movies in the retrieval system. He poses a query like ``What other sci-fi movies (besides Avengers: Endgame) were released in 2019?''
If the retrieval system cannot correctly address this exclusionary requirement, it can result in retrieval results containing content irrelevant to the user's interests (e.g., the movie ``Avengers: Endgame''), thus reducing user satisfaction.

Research on exclusionary retrieval remains relatively overlooked.
Early studies mainly focus on keyword-based methods~\cite{nakkouzi1990query, mcquire1998ambiguity, harvey2003challenge}.
The key idea is to construct boolean queries that include negation terms, which is essentially a post-processing strategy.
However, these methods exhibit limitations due to their reliance on structured queries, making them unsuitable for more diverse and complex natural language queries.
Although recent work has explored the impact of negation in modern retrieval models~\cite{rokach2008negation, koopman2010analysis, weller-etal-2024-nevir}, their focus is on comprehending the negation semantics within documents rather than the exclusionary nature of queries.
At present, there is no evaluation dataset to assess the capability of retrieval models in exclusionary retrieval.

To address this, our first contribution in this paper is the introduction of the resources for exclusionary retrieval, namely ExcluIR.
ExcluIR contains an evaluation benchmark to assess the capability of retrieval models in exclusionary retrieval, while also providing a training dataset that includes exclusionary queries.
The dataset is built based on HotpotQA~\cite{yang2018hotpotqa}.
We first use ChatGPT\footnote{\url{https://platform.openai.com/docs/models/gpt-3-5}} to generate an exclusionary query for two given relevant documents, requiring that only one document contains the answer while explicitly rejecting information from the other document.
Subsequently, we employ 17 workers to check each data instance in the benchmark to ensure data quality.
The training set comprises 70,293 exclusionary queries, while the benchmark includes 3,452 human-annotated exclusionary queries.
This dataset can evaluate whether retrieval models can correctly retrieve documents when dealing with exclusionary queries, providing a new perspective for evaluating retrieval models.

Our second contribution is to investigate the performance of existing retrieval methods with different architectures on exclusionary retrieval, including sparse retrieval~\cite{robertson2009probabilistic, nogueira2019doc2query}, dense retrieval~\cite{karpukhin2020dense, ni2022sentence}, and generative retrieval methods~\cite{bevilacqua2022autoregressive, wang2022neural}.
We conduct extensive experiments and have the following three main observations:
(1) Existing retrieval models with different architectures cannot fully understand the real intent of exclusionary queries;
(2) Generative retrieval models possess unique advantages in exclusionary retrieval, while late interaction models~\cite{khattab2020colbert, santhanam2022colbertv2} like ColBERT have obvious limitations in handling such queries;
(3) Fine-tuning the retrieval models with the training set of ExcluIR can improve the performance on exclusionary retrieval, but the results are still far from satisfactory.
We provide in-depth analyses of these observations.
These conclusions contribute valuable insights for future research on addressing the challenges of exclusionary retrieval.
\section{Dataset Construction}
As depicted in Figure~\ref{fig:constru}, the construction of the ExcluIR dataset involves the following steps:
(1) Initially, we extract document pairs from HotpotQA~\cite{yang2018hotpotqa}, where each data instance consisting of two interrelated documents;
(2) For each document pair, we employ ChatGPT to generate an exclusionary query.
(3) To enhance the diversity of the synthetic queries, we further use ChatGPT to rephrase them;
(4) Finally, to ensure the high quality of the dataset, we establish annotation guidelines and hire workers for manual correction.

\subsection{Collecting documents pairs}
We begin the process by collecting documents from the HotpotQA~\cite{yang2018hotpotqa} dataset, which is designed for multi-hop reasoning in question-answering task.
Each data instance includes two supporting documents that are interrelated. The model needs to comprehend the association between them and extract information from them to answer the question.
We extract two related documents from each data instance to form our document pairs.
In total, we collected 74,293 document pairs. After merging and removing duplicates, we obtained a document collection containing 90,406 documents.

\subsection{Generating exclusionary queries}
To construct our dataset efficiently, we design a prompt carefully to guide ChatGPT in generating exclusionary queries for each pair of documents (see Appendix~\ref{appendix:prompt}).
To ensure that the generated queries cover both positive and negative documents, we design a two-step construction strategy.
Specifically, we first instruct ChatGPT to generate a query relevant to both documents, and then guide ChatGPT to revise this query by adding a constraint to include the semantics of refusal to information from the negative document.

\subsection{Rewriting synthetic queries}
Although the prompt has been carefully adjusted, the generated queries often express the exclusionary phrases in a limited manner, such as ``excluding any information about'', ``except for any information'', and ``without referencing any information about.''
These expressions lack naturalness and deviate from real-world queries.
To increase the diversity and naturalness of the queries, we further instruct ChatGPT to rephrase them.
Then, we partition the ExcluIR dataset obtained in this step into training and test sets.
The test set is further manually corrected to construct the benchmark, which will be described in the following section.

\begin{figure}[t]
    \centering
    \includegraphics[scale=0.42]{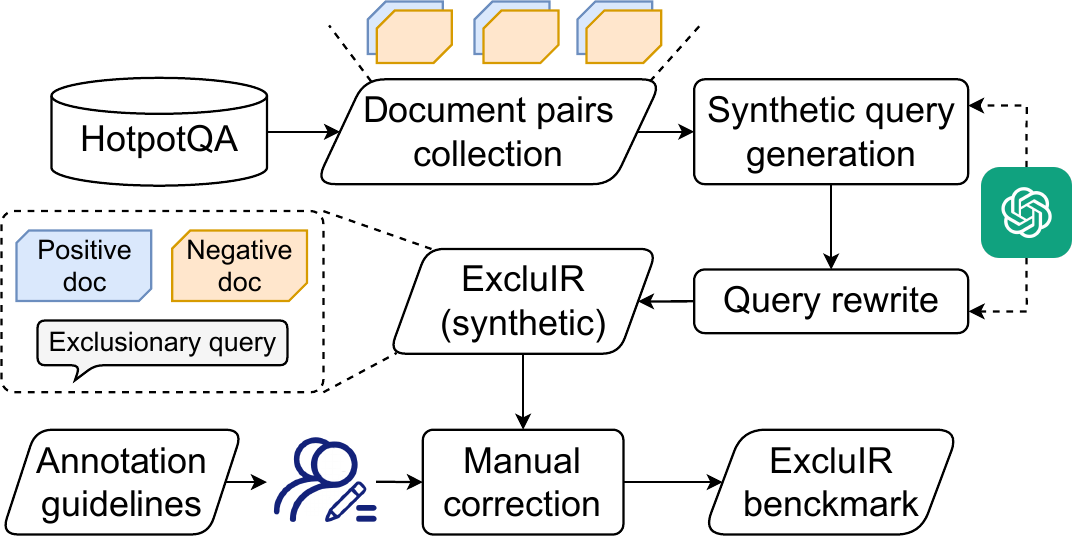}
    \caption{Overview of ExcluIR dataset construction process.}
    \label{fig:constru}
\vspace{-3mm}
\end{figure}

\subsection{Manually correcting data}
To build a reliable ExcluIR benchmark, we hire 17 workers for manual data correction.
We first sample 4,000 instances from the 74,293 data points obtained in the previous step.
Each instance contains two documents along with a synthetic query generated by ChatGPT.
We ask workers to check the synthetic exclusionary query to ensure its naturalness and correctness and they are encouraged to express the exclusionary nature of queries using diverse expressions.
The detailed requirements are provided in Appendix~\ref{appendix:requirements}.
To facilitate the correction process, we construct an online correction system.
In the system, we define three operations for workers to correct each data instance:

\begin{enumerate}[leftmargin=*, itemsep=0pt, topsep=2pt, label={(\arabic*)}]
    \item Criteria Met.
    If the synthetic query already meets the criteria, no further modifications are necessary.
    \item Query Modification.
    If the synthetic query fails to meet the criteria, modify or rewrite the query to align with the requirements.
    \item Discard Data.
    If it is difficult to write a query that meets the criteria based on these two documents, the workers can choose to discard the data.
\end{enumerate}

\subsection{Quality assurance}
We take some measures to ensure data quality:
First, we provide detailed documentation guidelines, including task definition, correction process, and specific criteria for exclusionary queries.
Second, we present multiple examples of exclusionary queries to help workers understand the task and requirements.
Third, we record a video to demonstrate the entire correction process and emphasize the key considerations that need special attention.
Fourth, we adopt a real-time feedback mechanism to allow workers to share the issues they encounter during the correction process.
We discuss these issues and provide solutions accordingly.
Finally, we randomly sample 10\% of the data of each worker for quality inspection.
If there are errors in the sampled data, we will ask the worker to correct the data again.

\begin{figure}[t]
    \centering
    \includegraphics[scale=0.265]{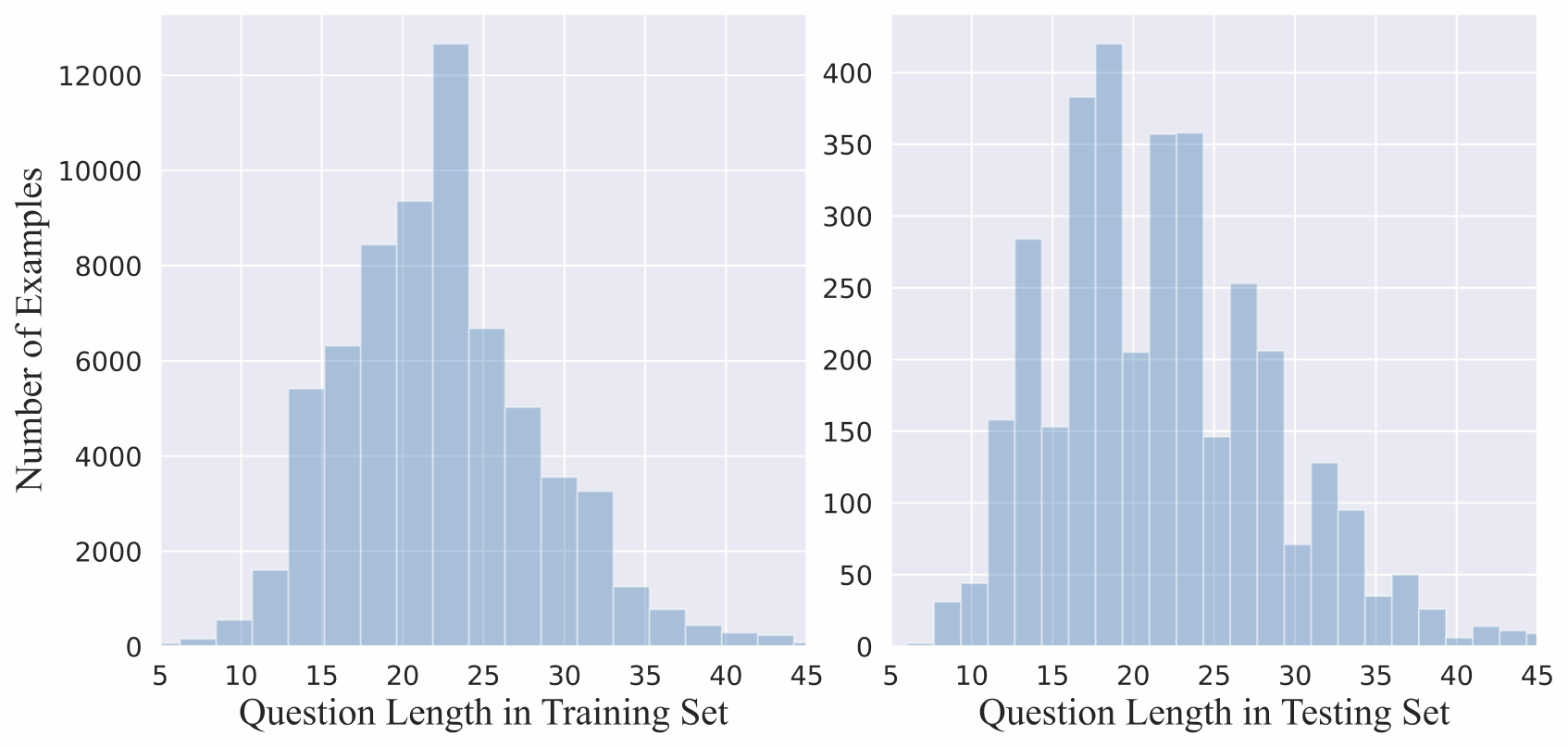}
    \caption{Distribution of the lengths of exclusionary queries in ExcluIR.}
    \label{fig:dataset}
\vspace{-3mm}
\end{figure}

\subsection{Dataset statistics}
Following the dataset construction process described above, we obtain 3,452 human-annotated entries for the benchmark and 70,293 exclusionary queries for the training set.
The average word counts for exclusionary queries in the training set and benchmark are 22.37 and 21.64, respectively.
To further investigate the diversity of data, we visualize the distribution of the lengths of exclusionary queries in Figure~\ref{fig:dataset}.
We show that the lengths of exclusionary queries are diverse, reflecting varying levels of complexity and details.
\section{Experimental Setups}
\subsection{Methods for comparison}
To evaluate the performance of various retrieval models on exclusionary retrieval, we select three types of retrieval models with different architectures: sparse retrieval, dense retrieval, and generative retrieval.

Sparse retrieval methods calculate the relevance score of documents using term matching metrics such as TF-IDF~\cite{robertson1997relevance}.
\begin{itemize}[leftmargin=*, itemsep=0pt, topsep=5pt]
\item \textbf{BM25}~\cite{robertson2009probabilistic} is a classical probabilistic retrieval method based on the normalization of the frequency of the term and the length of the document.
\item \textbf{DocT5Query}~\cite{nogueira2019doc2query} expands documents by generating pseudo queries using a fine-tuned T5 model before building the BM25 index~\cite{raffel2020exploring}.
\end{itemize}

Dense retrieval utilizes \acp{PLM} as the backbones to represent queries and documents as dense vectors for computing relevance scores.
\begin{itemize}[leftmargin=*, itemsep=0pt, topsep=5pt]
\item \textbf{DPR}~\cite{karpukhin2020dense} is a dense retrieval model based on dual-encoder architecture, which uses the representation of the [CLS] token of BERT~\cite{kenton2019bert}.
\item \textbf{Sentence-T5}~\cite{ni2022sentence} uses a fine-tuned T5 encoder model to encode queries and documents into dense vectors.
\item \textbf{GTR}~\cite{ni2022large} has the same architecture as Sentence-T5 and has been pretrained on two billion question-answer pairs collected from the Web.
\item \textbf{ColBERT}~\cite{khattab2020colbert} is a late interaction model that learns embeddings for each token in queries and documents, and then uses a MaxSim operator to calculate the relevance score.
\end{itemize}

Generative retrieval is an end-to-end retrieval paradigm.
\begin{itemize}[leftmargin=*, itemsep=0pt, topsep=5pt]
\item \textbf{GENRE}~\cite{de2020autoregressive} retrieves entities by generating their names through a seq-to-seq model, it can be applied to document retrieval by directly generating document titles.
The original GENRE is trained based on BART as the backbone, and we reproduce it using T5.
\item \textbf{SEAL}~\cite{bevilacqua2022autoregressive} retrieves documents by generating n-grams within them.
\item \textbf{NCI}~\cite{wang2022neural} proposes a prefix-aware weight-adaptive decoder architecture, leveraging semantic document identifiers and various data augmentation strategies like query generation.
\end{itemize}

\subsection{Evaluation metrics}
For the original test queries, we report the commonly used metrics: Recall at rank N~(R@N, N = 1,5,10) and Mean Reciprocal Rank at rank N~(MRR@N, N = 10).
Recall measures the proportion of relevant documents that are retrieved in the top N results.
MRR is the mean of the reciprocal of the rank of the first relevant document.

In ExcluIR, each exclusionary query $q$ has a positive document $d^+$ and a negative document $d^-$.
Thus, the difference between the rank of $d^+$ and the rank of $d^-$ can reflect the retrieval model's capability of comprehending the exclusionary query.
So we report $\mathrm{\Delta}$R@N and $\mathrm{\Delta}$MRR@N, which can be formulated as:
\begin{equation}
    \mathrm{\Delta}\text{R@N} = \text{R@N}(d^+) - \text{R@N}(d^-),
\end{equation}
\begin{equation}
    \mathrm{\Delta}\text{MRR@N} = \text{MRR@N}(d^+) - \text{MRR@N}(d^-).
\end{equation}
In addition, we also report \ac{RR}, which is the proportion of results where $d^+$ is ranked higher than $d^-$.
The expected value of \ac{RR} is 50\% with random ranking.

\begin{table*}[ht]
    \centering
    \caption{Performance of models trained on HotpotQA and tested on HotpotQA and ExcluIR.
    For the evaluation on HotpotQA, we report Recall@2 rather than Recall@1, since each query in HotpotQA has two supporting documents.}
    \vspace{-3mm}
    \renewcommand\arraystretch{1.1}
    \scalebox{0.93}{
    \begin{tabular}{llccccccccc}
    \toprule
        \multirow{2}{*}{Type} & \multirow{2}{*}{Model} & \multicolumn{4}{c}{HotpotQA} & \multicolumn{5}{c}{ExcluIR} \\ \cmidrule(r){3-6} \cmidrule(r){7-11}
        & & R@2 & R@5 & R@10 & MRR & R@1 & MRR & $\mathrm{\Delta}$R@1 & $\mathrm{\Delta}$MRR & RR \\ \hline
        \multirow{2}{*}{\makecell[l]{Sparse\\Retrieval}} & BM25 & 67.16  & 76.65  & 80.98  & 92.47  & 49.68  & 65.17  & 7.82  & 4.66  & 53.48 \\ 
        & DocT5Query & 69.19  & 77.88  & 81.65  & 94.10  & 50.98  & 67.50  & 7.85  & 3.81  & 53.85 \\
        \midrule
        \multirow{4}{*}{\makecell[l]{Dense\\Retrieval}} & DPR & 55.53  & 67.44  & 73.49  & 81.73  & 49.63  & 65.79  & 7.34  & 5.01  & 54.02 \\ 
        & Sentence-T5 & 57.63  & 68.45  & 74.29  & 82.48  & 51.04  & 66.27  & 10.11  & 7.01  & 55.41 \\ 
        & GTR & 61.82  & 73.57  & 79.42  & 85.50  & 54.87  & 70.88  & 14.40  & 8.79  & 57.42 \\ 
        & ColBERT & 73.58  & 83.73  & 87.95  & 94.44  & 54.00  & 71.24  & 10.72  & 6.42  & 55.57 \\
        \midrule
        \multirow{3}{*}{\makecell[l]{Generative\\Retrieval}} & GENRE & 48.87  & 51.67  & 53.24  & 75.25  & 48.03  & 63.22  & 4.35  & 0.13  & 52.10 \\ 
        & SEAL & 60.78  & 72.26  & 78.20  & 85.76  & 51.33  & 67.88  & 11.64  & 7.71  & 55.52 \\ 
        & NCI & 47.60  & 58.14  & 64.37  & 74.59  & 37.22  & 51.37  & 1.97  & 2.29  & 50.93 \\ \bottomrule
    \end{tabular}
    }
\label{tab1}
\vspace{-2mm}
\end{table*}

\begin{table*}[ht]
    \centering
    \renewcommand\arraystretch{1.1}
    \caption{Performance of models trained on NQ320k and tested on NQ320k and ExcluIR.}
    \vspace{-3mm}
    \scalebox{0.93}{
    \begin{tabular}{llccccccccc}
    \toprule
        \multirow{2}{*}{Type} & \multirow{2}{*}{Method} & \multicolumn{4}{c}{NQ320k} & \multicolumn{5}{c}{ExcluIR} \\ \cmidrule(r){3-6} \cmidrule(r){7-11}
        & & R@1 & R@5 & R@10 & MRR & R@1 & MRR & $\mathrm{\Delta}$R@1 & $\mathrm{\Delta}$MRR & RR \\ \hline
        \multirow{2}{*}{\makecell[l]{Sparse\\Retrieval}} & BM25 & 37.96  & 61.24  & 68.86  & 47.86  & 49.68  & 65.17  & 7.82  & 4.66  & 53.48 \\ 
        & DocT5Query & 42.63  & 66.18  & 73.38  & 52.69  & 50.98  & 67.50  & 7.85  & 3.81  & 53.85 \\
        \midrule
        \multirow{4}{*}{\makecell[l]{Dense\\Retrieval}} & DPR & 54.81  & 79.50  & 85.52  & 65.39  & 48.55  & 60.50  & 16.45  & 13.49  & 58.76 \\ 
        & Sentence-T5 & 59.63  & 82.78  & 87.42  & 69.57  & 57.76  & 66.34  & 32.90  & 27.96  & 67.83 \\ 
        & GTR & 62.35  & 84.67  & 89.17  & 71.90  & 59.79  & 69.00  & 34.85  & 28.12  & 68.31 \\ 
        & ColBERT & 60.08  & 84.19  & 89.41  & 70.50  & 57.01  & 70.88  & 20.02  & 15.26  & 59.97 \\ 
        \midrule
        \multirow{3}{*}{\makecell[l]{Generative\\Retrieval}} & GENRE & 56.25  & 71.21  & 74.00  & 62.80  & 31.63  & 37.63  & 11.44  & 10.15  & 58.65 \\ 
        & SEAL & 55.24  & 75.13  & 80.97  & 63.86  & 43.54  & 55.17  & 16.11  & 15.27  & 60.02 \\ 
        & NCI & 60.41  & 76.10  & 80.19  & 67.18  & 31.46  & 38.95  & 15.87  & 16.81  & 56.84 \\ \bottomrule
    \end{tabular}
    }
\label{tab2}
\vspace{-2mm}
\end{table*}

\subsection{Implementation details}
In our experiments, we use Elasticsearch to evaluate BM25 on both raw documents and the documents augmented with DocT5Query.
We train DPR and ColBERT using the bert-base-uncased architecture, train Sentence-T5, GENRE, and NCI using the t5-base architecture, and train SEAL using the BART-large architecture.
We reproduce NCI and SEAL by their official implementations and other methods are reproduced by our own implementations.
For query generation, we leverage the pre-trained model, DocT5Query~\cite{nogueira2019doc2query}, to generate pseudo queries for each document.
For the training of neural retrieval models, the max input length is set to 256 and the batch size is set to 32.
\section{Results and Analyses}
\begin{table*}[ht]
    \centering
    \caption{Performance of models after expanding the training data domain.
    NQ+H(Mix) indicates mixing the NQ320k and HotpotQA datasets for simultaneous training.
    NQ+H(Seq) indicates initial training on the NQ320k dataset followed by continual training on the HotpotQA dataset.}
    \vspace{-3mm}
    \renewcommand\arraystretch{1.1}
    \scalebox{0.93}{
    \begin{tabular}{llccccccccc}
    \toprule
        \multirow{2}{*}{Model} & \multirow{2}{*}{Training Set} & \multicolumn{4}{c}{HotpotQA} & \multicolumn{5}{c}{ExcluIR} \\ \cmidrule(r){3-6} \cmidrule(r){7-11}
        & & R@2 & R@5 & R@10 & MRR & R@1 & MRR & $\mathrm{\Delta}$R@1 & $\mathrm{\Delta}$MRR & RR \\ \hline
        \multirow{3}{*}{\makecell[l]{DPR}} & HotpotQA & 55.53  & 67.44  & 73.49  & 81.73  & 49.63  & 65.79  & 7.34  & 5.01  & 54.02  \\
        & NQ+H(Mix) & 53.19  & 65.05  & 71.52  & 79.57  & 48.93  & 64.47  & 6.95  & 4.59  & 53.94  \\
        & NQ+H(Seq) & \textbf{56.91}  & \textbf{69.02}  & \textbf{74.59}  & \textbf{82.74}  & \textbf{50.87}  & \textbf{67.12}  & \textbf{8.66}  & \textbf{5.99}  & \textbf{54.66}  \\
        \midrule
        \multirow{3}{*}{\makecell[l]{Sentence-T5}} & HotpotQA & 57.63  & 68.45  & 74.29  & 82.48  & 51.04  & 66.27  & 10.11  & 7.01  & 55.41 \\ 
        & NQ+H(Mix) & 54.32  & 65.67  & 72.02  & 79.56  & 51.45  & 66.58  & 11.27  & 8.71  & 56.10 \\
        & NQ+H(Seq) & \textbf{58.40}  & \textbf{69.05}  & \textbf{74.72}  & \textbf{82.66}  & \textbf{52.49}  & \textbf{67.82}  & \textbf{12.92}  & \textbf{9.44}  & \textbf{56.82} \\
        \midrule
        \multirow{3}{*}{\makecell[l]{ColBERT}} & HotpotQA & \textbf{73.58}  & \textbf{83.73}  & \textbf{87.95}  & 94.44  & \textbf{53.69}  & \textbf{70.82}  & \textbf{10.64}  & \textbf{6.35}  & \textbf{55.53} \\
        & NQ+H(Mix) & 71.54  & 82.46  & 86.40  & 94.58  & 52.78  & 69.91  & 8.86  & 5.21  & 54.49  \\
        & NQ+H(Seq) & 73.26  & 83.42  & 87.69  & \textbf{94.68}  & 51.27  & 69.21  & 5.82  & 2.10  & 52.93 \\
        \midrule
        \multirow{3}{*}{\makecell[l]{SEAL}} & HotpotQA & 60.78  & 72.26  & 78.20  & 85.76  & \textbf{51.33}  & \textbf{67.88}  & \textbf{11.64}  & \textbf{7.71}  & 55.52  \\
        & NQ+H(Mix) & \textbf{61.65}  & \textbf{72.80}  & \textbf{78.61}  & \textbf{86.39}  & 51.25  & 67.68  & 11.50  & 7.23  & \textbf{55.63}  \\
        & NQ+H(Seq) & 59.86  & 71.19  & 76.88  & 84.30  & 50.52  & 66.73  & 10.77  & 6.79  & 55.36  \\
        \bottomrule
    \end{tabular}
    }
\label{tab3}
\vspace{-2mm}
\end{table*}

In this section, we present six groups of experimental results and analyses to study:
(1) the performance of the existing retrieval models on ExcluIR (Section~\ref{main_results}),
(2) the strategy to improve the performance on ExcluIR, including expanding the training data domain (Section~\ref{add_nq}), incorporating our dataset into the training data (Section~\ref{add_our}), and increasing the size of the model (Section~\ref{model_size}),
(3) the explanation for the superiority of generative retrieval in ExcluIR (Section~\ref{gr_analysis}),
(4) the reason for the limitation of late interaction models like ColBERT in ExcluIR (Section~\ref{colbert_analysis}).

\subsection{How well do existing methods perform on ExcluIR?}
\label{main_results}
To evaluate the performance of various retrieval models trained on existing datasets in ExcluIR, we conduct our experiments on two well-known standard retrieval datasets: \ac{NQ}~\cite{kwiatkowski2019natural} and HotpotQA~\cite{yang2018hotpotqa}.
\ac{NQ} is a large-scale dataset for document retrieval and question answering.
The version we use is NQ320k, which consists of 320k query-document pairs.
HotpotQA is a question-answering dataset that focuses on multi-hop reasoning.
We split the original HotpotQA in the same way as our ExcluIR dataset, resulting in a 70k training set and a 3.5k test set.

The main performance of various methods on the ExcluIR benchmark and other test data are presented in Table~\ref{tab1} and \ref{tab2}.
We have the following observations from the results.

First, although these methods achieve good performance on the standard test data including HotpotQA and NQ320k, their performance on the ExcluIR benchmark is unsatisfactory.
Nearly all models score less than 10\% higher than random ranking on the RR metric.
Despite the Sentence-T5 and GTR models trained on NQ320k achieve the highest $\mathrm{\Delta}$R@1/$\mathrm{\Delta}$MRR/\ac{RR} scores, they are still far from achieving ideal performance.
This is attributed to the fact that negative documents are erroneously retrieved and ranked high, indicating that these models fail to comprehend the exclusionary nature of queries.

Second, sparse retrieval methods demonstrate a significant limitation in comprehending the exclusionary nature of queries, so they have almost no capability to handle ExcluIR.
As shown in Table~\ref{tab2}, the \ac{RR} scores of BM25 and DocT5Query are only 53.48\% and 53.85\%, which are only slightly higher than random ranking.
Their $\mathrm{\Delta}$R@1 and $\mathrm{\Delta}$MRR scores are lower than most neural retrieval models trained on NQ320k.
This is because these methods are based on the lexical match between queries and documents.
This limitation prevents them from focusing on the exclusionary phrases in the query, instead leading to a high relevance score for negative documents.

Third, the diversity of training data impacts the model's ability to comprehend exclusionary queries.
As can be seen from Table~\ref{tab1} and \ref{tab2}, the models trained on NQ320k exhibit better performance on ExcluIR than those trained on HotpotQA.
We believe this is because the queries in NQ320k are more diverse and contain more exclusionary queries.
Therefore, increasing the domain and diversity of training data can be beneficial for exclusionary retrieval.
We will conduct further experimental analysis in Section~\ref{add_nq}.

Furthermore, we also evaluate the performance of additional models trained on different datasets in ExcluIR.
Due to space constraints, these results are presented in Appendix~\ref{appendix:more_results}.

\subsection{How does expanding training data affect the performance?}
\label{add_nq}

To further understand the impact of training data on the performance in exclusionary retrieval, we select representative models from each category for additional experiments.
We extend the experiment in Table~\ref{tab1} by adding the NQ320k dataset to the training data.
We consider two settings for expanding training data: 
``Mix'' means mixing the two datasets for simultaneous training, and ``Seq'' means training on NQ320k with continual training on HotpotQA.
The results in Table~\ref{tab3} show that the impact of expanding the training data domain on ExcluIR varies across different models.
Specifically, we have the following observations.

For the bi-encoder models, including DPR and Sentence-T5, the ``Seq'' strategy results in improved performance on ExcluIR.
We consider that this is because the initial training on the NQ320k enhances the model's general comprehension capabilities, as evidenced by the improved performance on the HotpotQA test set.

However, expanding the training data does not help ColBERT and SEAL achieve better results on ExcluIR.
While ColBERT exhibits competitive performance on two standard datasets, its performance diminishes on ExcluIR.
This is because ColBERT calculates the document relevance score based on token-level matching, leading to the overlooking of exclusionary phrases in queries, which is crucial for exclusionary retrieval.
We visualize the relevance calculation of ColBERT to further understand its performance in Section~\ref{colbert_analysis}.
As for SEAL, the inherent limitation of generative retrieval models in poorly generalizing to new or out-of-distribution documents explains why expanding the training data does not lead to improved performance on ExcluIR~\cite{lee2023nonparametric, mehta2023dsi++}.

Overall, expanding training data does not stably enhance the performance of models on ExcluIR.
We consider the primary reason to be the lack of exclusionary queries in the training data.
Therefore, in the next section, we will investigate the impact of incorporating our training set which consists of exclusionary queries into the training data.

\subsection{How does incorporating our dataset into training data affect the performance?}
\label{add_our}

\begin{figure}[t]
    \centering
    \subfigure{
        \centering
        \includegraphics[width=0.98\linewidth]{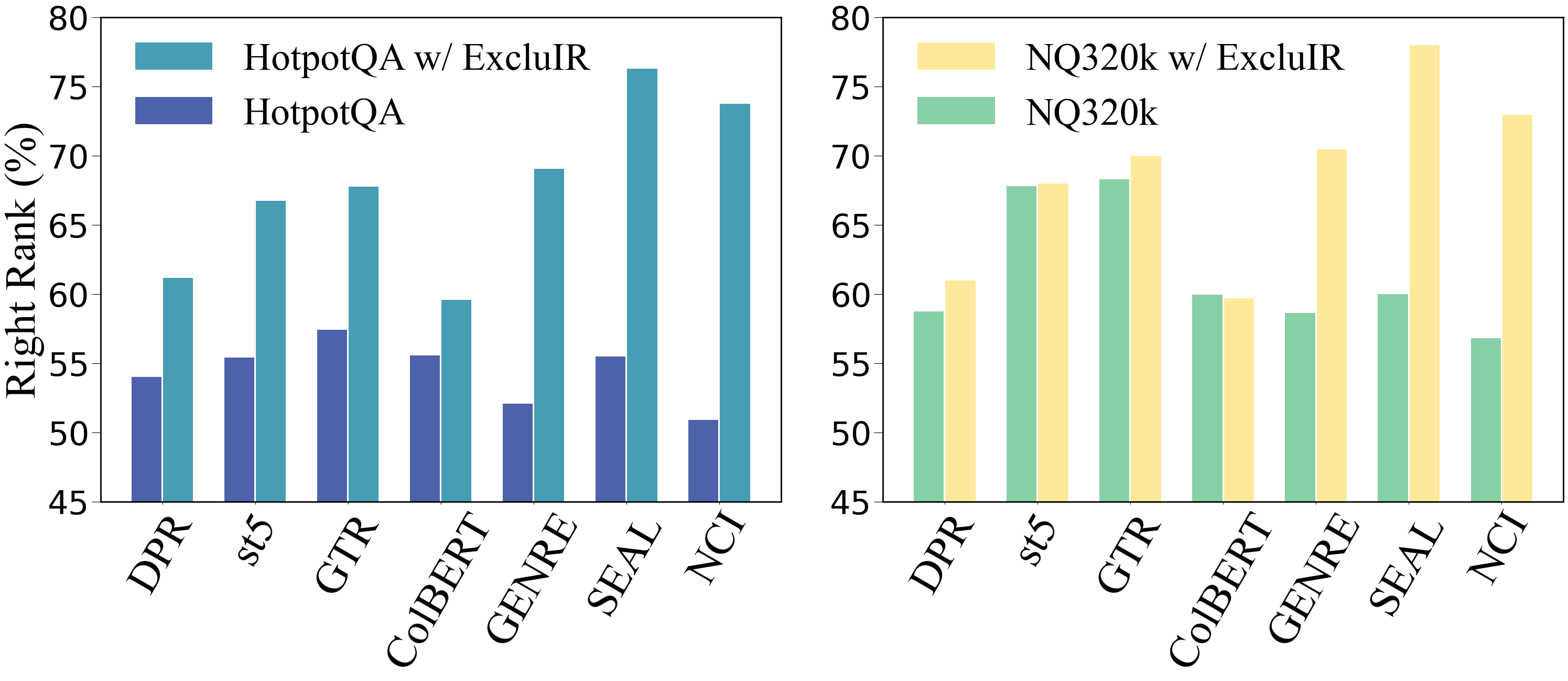}}
    \vspace{-4mm}
    \subfigure{
        \centering
        \includegraphics[width=0.98\linewidth]{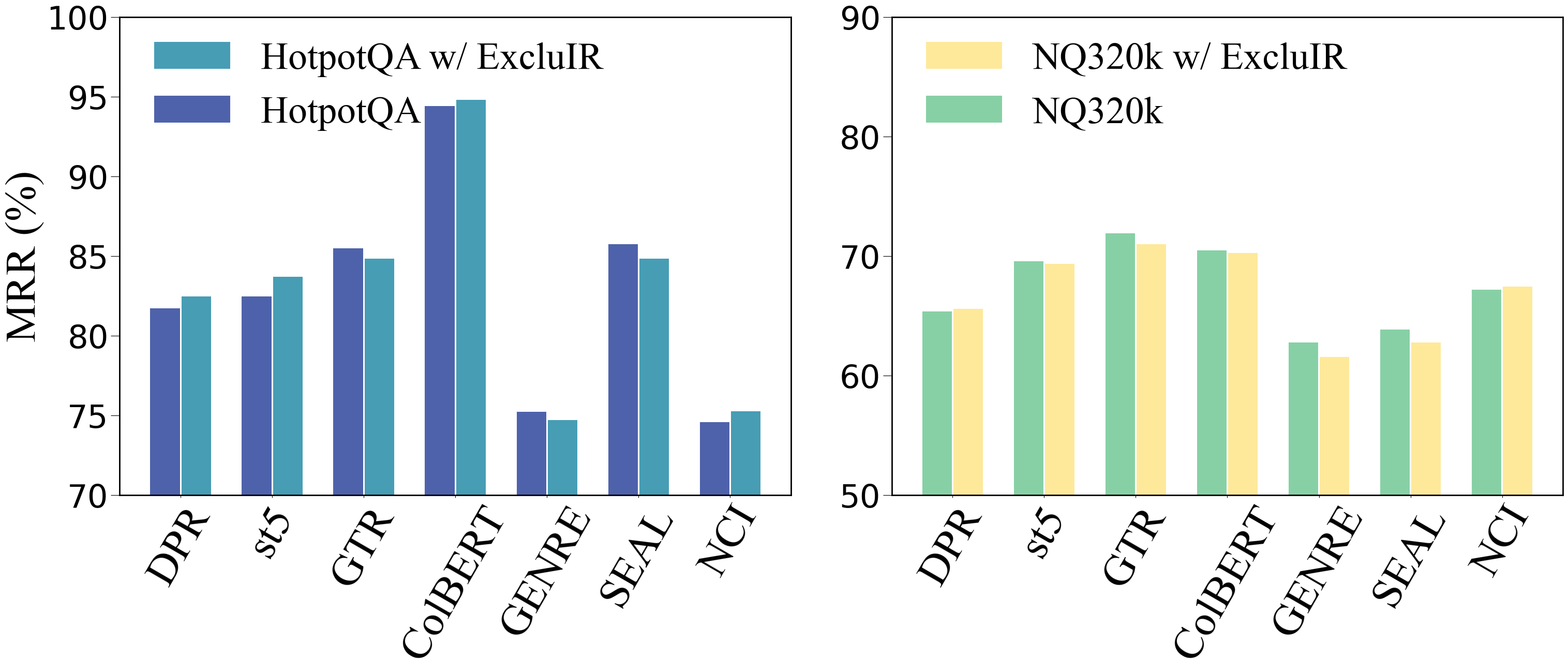}}
    \caption{Performance of models under different training data settings.
    The upper figures show the \ac{RR} score of various models on the ExcluIR benchmark, and the lower figures show the performance of these models on HotpotQA and NQ320k.
    The different colors of the bars represent different training data.
    Full results are presented in Appendix~\ref{appendix:aug_results}.}
    \label{fig:ourdata}
\vspace{-3mm}
\end{figure}

Previous experiments have demonstrated that models trained on HotpotQA and NQ320k perform unsatisfactorily on ExcluIR.
We consider this is partly due to the lack of exclusionary queries in the training data.
Therefore, in this section, we incorporate the ExcluIR training set into the training data to assess its impact on performance.
From the results in Figure~\ref{fig:ourdata}, we have three main observations.

First, merging the ExcluIR training set into the training data can significantly enhance the model's ability to comprehend exclusionary queries.
For instance, with NQ320k as the original dataset, SEAL achieves 18\% improvement (60.02\% vs. 78.02\%) in RR by integrating the ExcluIR training set, with only 1.08\% decrease (63.86\% vs. 62.78\%) in performance of original test data.
This is because the ExcluIR training set contains a large number of exclusionary queries, which can help the retrieval model to better comprehend the exclusionary nature of queries.

Second, when training data contain exclusionary queries, generative retrieval methods are more adept at learning the exclusionary nature of queries compared to dense retrieval methods.
As shown in Figure~\ref{fig:ourdata}, although dense retrieval models trained on two original datasets perform better on ExcluIR, the augmentation of the ExcluIR training set leads to a greater improvement in generative retrieval models, ultimately surpassing dense retrieval methods overall.
On average, generative retrieval models, including GENRE, SEAL, and NCI, achieve a 17.75\% improvement, in contrast to the average 4.77\% improvement observed in dense retrieval models.
This is because the generative retrieval model is more suitable for capturing the complex relationships between queries and documents in terms of model architecture and training objectives.
We present a more detailed analysis in Section~\ref{gr_analysis}.

Third, consistent with the conclusion in Section~\ref{add_nq}, ColBERT fails to achieve satisfactory performance, even after fine-tuning on ExcluIR.
As demonstrated in Figure~\ref{fig:ourdata}, among the models trained with the ExcluIR training set, ColBERT exhibits the lowest performance, with \ac{RR} score of 59.59\% on HotpotQA w/ ExcluIR and 59.71\% on NQ320k w/ ExcluIR.
As mentioned in Section~\ref{add_nq}, the relevance score calculation method utilized by ColBERT is not conducive to handling exclusionary queries.
We will provide a more detailed analysis in Section~\ref{colbert_analysis}.

\subsection{How does the model size affect the performance?}
\label{model_size}
To analyze the impact of model size on the performance of ExcluIR, we increase model sizes of DPR, sentence-t5, GENRE, NCI and train them on different datasets.
Specifically, for DPR, we utilize two variants: bert-base-uncased and bert-large-uncased.
For sentence-t5, GENRE and NCI, we adopt t5-base and t5-large.

The results are presented in Table~\ref{tab4}.
We note that increasing the model size improves performance on ExcluIR when the training data includes exclusionary queries.
This is consistent with the observations in \cite{ravichander2022condaqa}, which shows that larger models are better at understanding the implications of negated statements in documents.

However, when training on two original datasets, increasing the model size does not always lead to improved performance on ExcluIR.
The results in Table~\ref{more_results} support this observation.
For example, the performance of monot5-base is inferior to monot5-large, but the performance of stsb-roberta-large decreases significantly compared to stsb-roberta-base.
This indicates that simply increasing model size cannot solve the challenges of exclusionary retrieval, we should investigate building more training data and proposing new training strategies.

\begin{table*}[!ht]
    \centering
    \caption{Performance with different model sizes on ExcluIR.
    For DPR, the base version is bert-base-uncased, and the large version is bert-large-uncased.
    For sentence-t5, GENRE and NCI, the base version is t5-base, and the large version is t5-large.
    $\textcolor{teal}{\uparrow}$ indicates that an increase in model size improves performance, while $\textcolor{purple}{\downarrow}$ indicates the opposite.}
    \vspace{-3mm}
    \renewcommand\arraystretch{1.1}
    \scalebox{0.93}{
    \begin{tabular}{llccccccl}
    \toprule
        \multirow{2}{*}{\makecell[l]{Training\\set}} & \multirow{2}{*}{Model} & \multicolumn{3}{c}{Base} & \multicolumn{3}{c}{Large} \\ \cmidrule(r){3-5} \cmidrule(r){6-8}
        & & $\mathrm{\Delta}$R@1 & $\mathrm{\Delta}$MRR & RR & $\mathrm{\Delta}$R@1 & $\mathrm{\Delta}$MRR & RR \\ \cmidrule(r){1-8}
        \multirow{4}{*}{\makecell[l]{HotpotQA}} & DPR & 7.34  & 5.01  & 54.02  & 8.00  & 6.22  & 54.25 & $\textcolor{teal}{\uparrow}$ \\
        & Sentence-T5 & 10.11  & 7.01  & 55.41  & 7.21  & 5.23  & 53.78 & $\textcolor{purple}{\downarrow}$ \\
        & GENRE & 4.35  & 0.13  & 52.10  & -1.71  & -3.09  & 49.01 & $\textcolor{purple}{\downarrow}$ \\ 
        & NCI & 1.97 & 2.29 & 50.93  & 1.05  & 1.41  & 50.64 & $\textcolor{purple}{\downarrow}$ \\  \cmidrule(r){1-8}
        \multirow{4}{*}{\makecell[l]{HotpotQA\\w/ ExcluIR}} & DPR & 21.32  & 14.93  & 61.19  & 24.55  & 17.88  & 62.63 & $\textcolor{teal}{\uparrow}$ \\
        & Sentence-T5 & 33.78  & 24.49  & 66.75  & 37.05  & 26.50  & 69.01 & $\textcolor{teal}{\uparrow}$ \\
        & GENRE & 38.71  & 18.34  & 69.07  & 42.15  & 20.20  & 70.96 & $\textcolor{teal}{\uparrow}$ \\
        & NCI & 42.29 & 38.38 & 73.75 & 43.74 & 38.56 & 73.61 & $\textcolor{teal}{\uparrow}$ \\ \cmidrule(r){1-8}
        \multirow{4}{*}{\makecell[l]{NQ320k}} & DPR & 16.45  & 13.49  & 58.76  & 20.83  & 17.16  & 61.62 & $\textcolor{teal}{\uparrow}$ \\
        & Sentence-T5 & 32.90  & 27.96  & 67.83  & 34.36  & 29.94  & 69.02 & $\textcolor{teal}{\uparrow}$ \\
        & GENRE & 11.44  & 10.15  & 58.65  & 11.03  & 8.88  & 55.82 & $\textcolor{purple}{\downarrow}$ \\
        & NCI & 15.87 & 16.81 & 56.84 & 21.27 & 22.86 & 62.54 & $\textcolor{teal}{\uparrow}$ \\ \cmidrule(r){1-8}
        \multirow{4}{*}{\makecell[l]{NQ320k\\w/ ExcluIR}} & DPR & 21.52  & 16.38  & 61.00  & 25.52  & 19.15  & 63.47 & $\textcolor{teal}{\uparrow}$ \\
        & Sentence-T5 & 34.47  & 26.19  & 68.00  & 37.34  & 28.70  & 69.65 & $\textcolor{teal}{\uparrow}$ \\
        & GENRE & 41.19  & 20.31  & 70.48  & 46.04  & 23.24  & 72.86 & $\textcolor{teal}{\uparrow}$ \\
        & NCI & 41.13 & 39.92 & 72.97 & 43.13 & 41.86 & 74.45 & $\textcolor{teal}{\uparrow}$ \\ \bottomrule
    \end{tabular}
    }
\label{tab4}
\vspace{-2mm}
\end{table*}

\begin{figure}[t]
    \centering
    \includegraphics[scale=0.49]{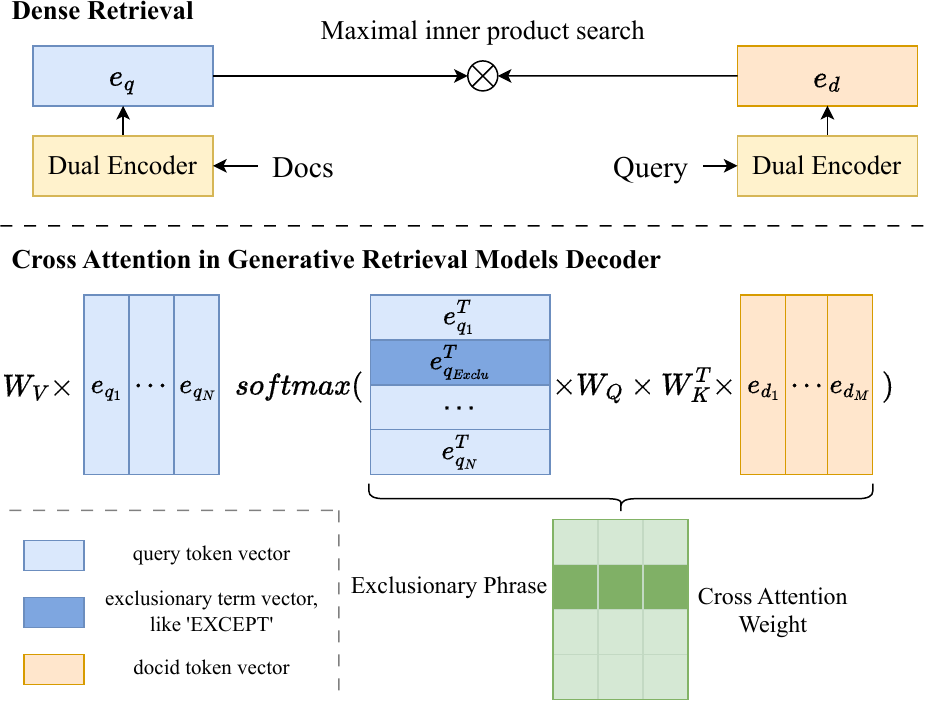}
    \caption{Summary of the analysis that shows the differences between dense retrieval and generative retrieval models in handling ExcluIR.}
    \label{fig:gr}
\vspace{-2mm}
\end{figure}

\subsection{Why are generative retrieval models superior in ExcluIR?}
\label{gr_analysis}

Generative retrieval models have inherent advantages in comprehending exclusionary queries.
We try to analyze and explain the reason from the architecture of generative models.

First, as a comparison, we show that bi-encoder models have a representation bottleneck for exclusionary queries.
When two documents are similar but have some differences that the user would like to distinguish, it is difficult to ensure that the vector representation of the query remains distant from the negative document while closely aligning with the positive document. 
This representation bottleneck prevents the model from correctly comprehending the true intent of the query.
We present this proof in Appendix~\ref{appendix:proof_bi}.

Generative retrieval models adopt a sequence-to-sequence framework, such as T5 or BART, which estimates the probability of generating the document IDs given the query using a conditional probability model: $P(d|q)$.
When generating document IDs, multiple cross-attention layers in the decoder can capture the token-level semantic information in the query, a phenomenon also explored by~\citet{wu2024generative}.
Assuming the decoder consists of $L$ layers, for the $l$-th layer ($0\leq l< L$), the cross-attention layer is given by:
\begin{equation}
    S^{(l+1)} = \text{softmax}\left(\frac{Q^{(l)}K^{(l)T}}{\sqrt{d_k}}\right)V^{(l)},
\end{equation}
where $Q^{(l)} = W^{(l)}_qS^{(l)}$, $K^{(l)} = W^{(l)}_kH^{(l)}_{q}$, $V^{(l)} = W^{(l)}_vH^{(l)}_{q}$, and $H^{(l)}_{q} = \left[e_{q_1},\cdots,e_{q_N}\right]$ are query token vectors generated by encoder, $S^{(l)} = \left[e_{d_1},\cdots,e_{d_M}\right] $ are generated embedding vectors for docid tokens at $l$-th layer, $W^{(l)}_q$, $W^{(l)}_k$ and $W^{(l)}_v$ are learnable cross-attention weight matrices.
We visualize the cross attention in generative models to summarize our analysis.
As shown in Figure~\ref{fig:gr}, the multi-level cross-attention mechanism allows the model to focus intensively on key terms in the query, including exclusionary phrases (highlighted in dark green).
Thus, even when faced with queries with complex semantics, generative retrieval models are capable of effectively capturing the query intent.

\subsection{Why does ColBERT underperform in ExcluIR?}
\label{colbert_analysis}

\begin{figure*}[t]
    \centering
    \includegraphics[scale=0.335]{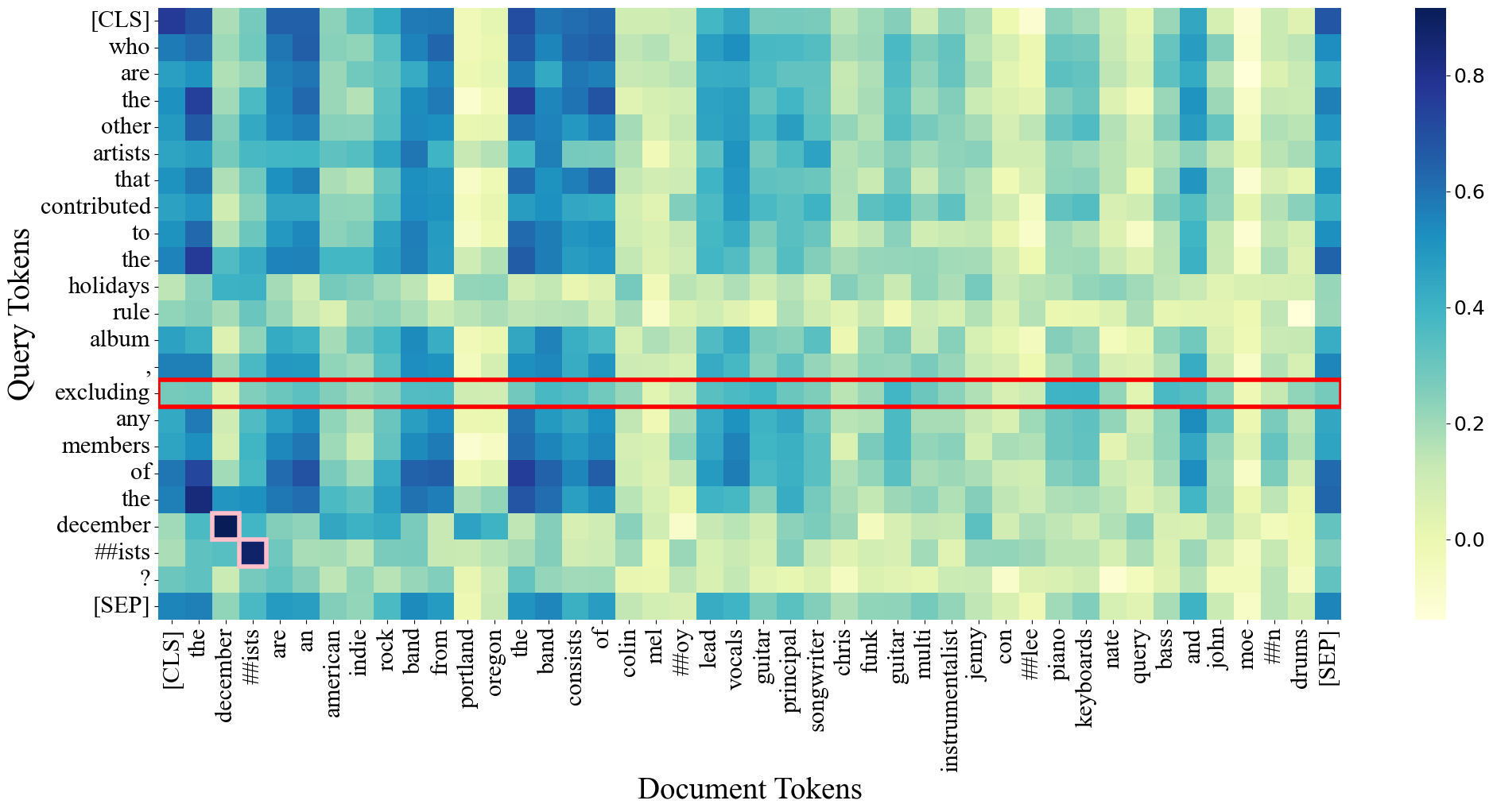}
    \caption{A relevance calculation visualization between query and negative document of ColBERT.
    Each value in the heatmap represents the result of the dot product between the query token vector and the negative document token vector.
    The red highlight indicates the relevance of the token `excluding' in the query to each token in the negative document,
    and the pink highlights indicate the token with the highest relevance score.}
    \label{fig5}
\end{figure*}

Late interaction models like ColBERT struggle to comprehend the exclusionary nature of queries.
From the previous experimental results, we can see that ColBERT performs worse than other neural retrieval models in ExcluIR.
As ColBERT introduces a late interaction architecture, it calculates document relevance scores based on the matching of token-level vectors between queries and documents.
Consequently, the exclusionary phrases in queries pose a challenge for matching with document tokens.

As we can see in Figure~\ref{fig5}, the token `exclude' in the query exhibits relatively low relevance with every token in the negative document.
This indicates that ColBERT barely comprehends the true intent of the query.
And we can also notice that `decemberists' appears both in the query and negative document, contributing a very high relevance score, which is disadvantageous for exclusionary retrieval.
Although the `decemberists' band is mentioned in the query, the intent of the query is to avoid retrieving information about this band.
Therefore, ColBERT inherently lacks the capability to comprehend the queries with complex intentions, limiting its effectiveness in ExcluIR.
We present more cases in Appendix~\ref{appendix:colbert_case}.
\section{Related Work}
Early studies in exclusionary retrieval primarily focus on keyword-based methods.
These approaches typically treat user queries as logical expressions of boolean operations~\cite{nakkouzi1990query, strzalkowski1995natural, mcquire1998ambiguity, harvey2003challenge}.
However, these methods depend on explicit and deterministic rules, lack the flexibility to handle subtle and conditional exclusions, and are unsuitable for more realistic retrieval scenarios.
In addition, there is a similar retrieval task, known as argument retrieval~\cite{wachsmuth2018retrieval}.
This task aims to retrieve the best counterargument for a given argument on any controversial topic.
But argument retrieval implicitly requires the model to find the counterargument to query, the intention of exclusion is not explicitly expressed in the query.
\citet{wang2022learn} first investigate exclusionary retrieval in \ac{T2VR}.
They demonstrate that existing video retrieval models performed poorly when dealing with queries like ``find shots of kids sitting on the floor and not playing with the dog.''
To the best of our knowledge, there has been no research on exclusionary retrieval in document retrieval.

Another recent work by \citet{weller-etal-2024-nevir} introduces NevIR, a benchmark designed to assess the capability of neural information retrieval systems in handling negation.
NevIR requires retrieval models to rank two documents that differ only in negation, where both documents remain consistent in all other aspects except the key negation.
Similarly, \citet{rokach2008negation, koopman2010analysis} investigate the impact of negation contexts within documents on retrieval performance.
For example, a search for ``headache'' might retrieve patient records containing ``the patient has no symptoms of headache.''
Differently, we focus on exclusionary retrieval, studying whether the retrieval model can comprehend the intent of exclusionary queries.
\section{Conclusion}
In this work, we focus on a common yet insufficiently studied retrieval scenario called exclusionary retrieval, where users explicitly express the information they do not want to obtain.
We have provided the community with a new benchmark, named ExcluIR, which focuses on exclusionary queries that explicitly express the information users do not want to obtain.
We have conducted extensive experiments, which demonstrate that existing retrieval methods with different architectures perform poorly on ExcluIR.
Notably, ExcluIR cannot be solved by simply adding training data domains or increasing model sizes.
Additionally, our analyses indicate that generative retrieval models inherently excel at comprehending exclusionary queries compared with sparse and dense retrieval models.
We hope that this work can inspire future research on ExcluIR.
\section*{Limitations and Future Work}
This work has the following limitations.
First, although the training data we build can significantly improve the performance of various retrieval models on ExcluIR, there is still a considerable gap from human performance (with \ac{RR} score of 100\%).
In future work, we plan to investigate how to make use of the advantages of \ac{GR} to further improve the ability of retrieval models in exclusionary retrieval.
Second, in practical retrieval scenarios, the exclusions in the query can be expressed in different ways.
Some are directly stated within a single-round query, while others are implied within the context of multi-round queries.
For example, users might prefer that the results of the current query do not include content retrieved in previous rounds, even though this intent of exclusion is not directly expressed within the query.
In this work, we have only considered the former scenario, further research is required to explore a broader range of exclusionary retrieval scenarios.
\section*{Ethical Considerations}
We realize the potential risks in the research of ExcluIR, thus, it is necessary to pay attention to the ethical issues.
All raw data collected in this study are sourced from publicly available datasets, with ethical considerations approved by publishers.
In the process of data annotation, all workers are informed of the research objectives in advance.
We did not collect any personal privacy information and all data used in our research is obtained following legal and ethical standards.
\bibliography{anthology,references}
\newpage
\appendix
\section{Prompt templates}
\label{appendix:prompt}
We present the prompt that used to guide ChatGPT generating and rephrase the exclusionary query for each pair of documents in Table~\ref{tab:prompt}.

\section{Requirements for manual correction}
\label{appendix:requirements}
During manual correction, to ensure the quality of data, we provide the following requirements for workers.
\begin{enumerate}[leftmargin=*, itemsep=0pt, topsep=2pt]
    \item The query should be relevant to the positive document.
    \item The query should include an exclusionary constraint to clearly refuse to inquire about the information in the negative document.
    \item The query should contain enough information, not just using a person's name to represent a document.
    \item You should use diverse expressions to express the exclusionary constraint, rather than repetitively using the same terms like `excluding.'
\end{enumerate}

\begin{table*}[ht]
    \centering
    \renewcommand{\arraystretch}{1.2}
    \caption{Prompt templates for query generation and rephrasing.}
    \label{tab:prompt}
    \vspace{-3mm}
    \begin{tabular}{p{0.15\linewidth} p{0.7\linewidth}}
    \toprule
    Task & Prompt template \\
    \hline
    Generation & You will be provided with two documents, and you need to:
    \begin{enumerate}[itemsep=0pt, topsep=2pt]
        \item generate a query that is relevant to both Document 1 and Document 2; and
        \item revise this query to include a constraint or condition that makes it explicitly refuse to inquire about any information in Document 1.
    \end{enumerate}
    Reply Format: \newline
    Query: \newline
    Revised Query: \\
    \hline
    Rephrasing & Rephrase the following query to make it smoother, more reasonable, more natural and more realistic.
    Do not answer this query but just polish it.
    You should make this query more like a real human query, but do not change the semantics of this query.
    \newline\newline
    Query:\newline
    \textit{raw query}
    \newline\newline
    Rewritten Query: \\
    \bottomrule
    \end{tabular}
\vspace{-2mm}
\end{table*}

\section{The experimental results of additional models on ExcluIR}
\label{appendix:more_results}
We present more results in Table~\ref{more_results} showing the performance of various models in ExcluIR.
Most of the models are from sentence-transformers~\cite{reimers2019sentence}, except for RocketQA~\cite{qu2021rocketqa, ren2021rocketqav2} and monot5~\cite{nogueira2020document}.
Since cross-encoder models are used for re-ranking, it is very time-consuming to calculate the relevance score of all documents in the corpus.
Therefore, We first retrieve the top 100 documents using BM25, and then re-rank them.
We find that the Recall@100 of BM25 for positive and negative documents is 95.77\% and 94.74\%, so this strategy can ensure fairness.

\begin{table*}[ht]
    \centering
    \caption{The performance of various models on ExcluIR.
    Training Data indicates the source of training data for the model, and Params indicates the number of parameters in the model.}
    \label{more_results}
    \vspace{-3mm}
    \renewcommand\arraystretch{1.1}
    \begin{tabular}{llllccc}
    \toprule
        Type & Training Data & Params & Model & $\mathrm{\Delta}$R@1 & $\mathrm{\Delta}$MRR & RR \\ \hline
        \multirow{11}{*}{\makecell[l]{Bi-\\Encoders}} & MSMarco & 218M & RocketQA v1 & 31.62 & 24.94 & 65.13 \\
        & NQ & 218M & RocketQA v1 & 25.09 & 21.47 & 61.31 \\
        & NQ & 218M & RocketQA v2 & 17.93 & 15.74 & 53.61 \\
        & Multi-Datasets & 23M & all-MiniLM-L6-v2 & 26.41 & 19.42 & 62.09 \\ 
        & Multi-Datasets & 33M & all-MiniLM-L12-v2 & 27.62 & 21.05 & 63.29 \\
        & Multi-Datasets & 109M & all-mpnet-base-v2 & \textbf{39.32} & \textbf{32.01} & \textbf{69.04} \\
        & Multi-Datasets & 82M & all-distilroberta-v1 & 37.56 & 27.63 & 67.98 \\
        & Multi-Datasets & 66M & multi-qa-distilbert-cos-v1 & 25.90 & 18.42 & 61.77 \\
        & Multi-Datasets & 109M & multi-qa-mpnet-base-dot-v1 & 37.80 & 29.04 & 68.41 \\
        & Multi-Datasets & 23M & multi-qa-MiniLM-L6-cos-v1 & 24.11 & 17.87 & 60.97 \\
        & Multi-Datasets & 278M & \makecell[l]{paraphrase-multilingual-\\mpnet-base-v2} & 33.72 & 27.61 & 65.85 \\ \hline
        \multirow{11}{*}{\makecell[l]{Cross-\\Encoders}} & MSMarco & 23M & ms-marco-MiniLM-L-6-v2 & 27.56 & 16.61 & 63.35\\
        & MSMarco & 33M & ms-marco-MiniLM-L-12-v2 & 27.08 & 16.47 & 63.12 \\
        & SQuAD & 109M & qnli-electra-base & 23.60 & \textbf{27.76} & 53.87 \\
        & STSB & 125M & stsb-roberta-base & 13.48 & 15.13 & 59.38 \\
        & STSB & 355M & stsb-roberta-large & 6.50 & 8.26 & 50.27 \\
        & MSMarco & 223M & monot5-base-msmarco-10k & 32.54 & 18.87 & 65.85 \\
        & MSMarco & 738M & monot5-large-msmarco-10k & \textbf{42.80} & 23.71 & \textbf{70.91} \\
        & MSMarco & 2852M & monot5-3b-msmarco-10k & 42.17 & 23.35 & 70.74 \\
        & MSMarco & 109M & RocketQA-v2\_marco\_ce & 37.22 & 21.11 & 68.24 \\
        & MSMarco & 335M & RocketQA-v1\_marco\_ce & 40.39 & 22.40 & 70.02 \\
        & NQ & 335M & RocketQA-v1\_nq\_ce & 41.56 & 22.98 & 70.48 \\
    \bottomrule
    \end{tabular}
\vspace{-2mm}
\end{table*}

\section{The complete results of training with ExcluIR on HotpotQA and NQ320k}
\label{appendix:aug_results}

Table~\ref{aug_results_h} and \ref{aug_results_n} show full results of retrieval models performance after augmenting the HotpotQA and NQ320k with the ExcluIR training set, respectively.

\begin{table*}[ht]
    \centering
    \caption{The complete results of the impact of augmenting HotpotQA with ExcluIR training set.}
    \vspace{-3mm}
    \renewcommand\arraystretch{1.1}
    \scalebox{0.93}{
    \begin{tabular}{llccccccccc}
    \toprule
        \multirow{2}{*}{Model} & \multirow{2}{*}{Training Data} & \multicolumn{4}{c}{HotpotQA} & \multicolumn{5}{c}{ExcluIR} \\ \cmidrule(r){3-6} \cmidrule(r){7-11}
        & & R@2 & R@5 & R@10 & MRR & R@1 & MRR & $\mathrm{\Delta}$R@1 & $\mathrm{\Delta}$MRR & RR \\ \hline
        \multirow{2}{*}{DPR} & HotpotQA & 55.53  & 67.44  & 73.49  & 81.73  & 49.63  & 65.79  & 7.34  & 5.01  & 54.02 \\ 
        & H. w/ ExcluIR & \textbf{58.26}  & \textbf{70.48}  & \textbf{76.81}  & \textbf{83.60}  & \textbf{59.30}  & \textbf{73.20}  & \textbf{24.45}  & \textbf{17.88}  & \textbf{62.63} \\
        \midrule
        \multirow{2}{*}{Sentence-T5} & HotpotQA & 57.63  & 68.45  & 74.29  & 82.48  & 51.04  & 66.27  & 10.11  & 7.01  & 55.41 \\ 
        & H. w/ ExcluIR & \textbf{58.65}  & \textbf{69.60}  & \textbf{75.48}  & \textbf{83.72}  & \textbf{63.73}  & \textbf{75.85}  & \textbf{33.78}  & \textbf{24.49}  & \textbf{66.75} \\
        \midrule
        \multirow{2}{*}{GTR} & HotpotQA & 61.82  & 73.57  & 79.42  & \textbf{85.50}  & 54.87  & 70.88  & 14.40  & 8.79  & 57.42  \\ 
        & H. w/ ExcluIR & \textbf{61.99}  & \textbf{73.83}  & \textbf{79.45}  & 84.86  & \textbf{64.98}  & \textbf{77.75}  & \textbf{34.85}  & \textbf{23.85}  & \textbf{67.79} \\
        \midrule
        \multirow{2}{*}{ColBERT} & HotpotQA & \textbf{73.58}  & \textbf{83.73}  & \textbf{87.95}  & 94.44  & 54.00  & 71.24  & 10.72  & 6.42  & 55.57 \\ 
        & H. w/ ExcluIR & 72.90  & 83.26  & 87.50  & \textbf{94.80}  & \textbf{58.14}  & \textbf{74.95}  & \textbf{18.80}  & \textbf{12.74}  & \textbf{59.59} \\
        \midrule
        \multirow{2}{*}{GENRE} & HotpotQA & \textbf{48.87}  & \textbf{51.67}  & \textbf{53.24}  & \textbf{75.25}  & 48.03  & 63.22  & 4.35  & 0.13  & 52.10 \\ 
        & H. w/ ExcluIR & 48.60  & 51.26  & 53.03  & 74.71  & \textbf{64.98}  & \textbf{72.54}  & \textbf{38.71}  & \textbf{18.34}  & \textbf{69.07}  \\
        \midrule
        \multirow{2}{*}{SEAL} & HotpotQA & \textbf{60.78}  & 72.26  & \textbf{78.20}  & \textbf{85.76}  & 51.33  & 67.88  & 11.64  & 7.71  & 55.52 \\ 
        & H. w/ ExcluIR & 60.34  & \textbf{72.39}  & 77.97  & 84.85  & \textbf{69.03}  & \textbf{78.66}  & \textbf{48.95}  & \textbf{39.55}  & \textbf{76.29}\\
        \midrule
        \multirow{2}{*}{NCI} & HotpotQA & 47.60  & 58.14  & 64.37  & 74.59  & 37.22  & 51.37  & 1.97  & 2.29  & 50.93 \\ 
        & H. w/ ExcluIR & \textbf{47.80}  & \textbf{59.15}  & \textbf{64.75}  & \textbf{75.28}  & \textbf{59.76}  & \textbf{68.90}  & \textbf{42.29}  & \textbf{38.38}  & \textbf{73.75} \\
        \bottomrule
    \end{tabular}
    }
\label{aug_results_h}
\vspace{-2mm}
\end{table*}

\begin{table*}[ht]
    \centering
    \caption{The complete results of the impact of augmenting NQ320k with ExcluIR training set}
    \vspace{-3mm}
    \renewcommand\arraystretch{1.1}
    \scalebox{0.93}{
    \begin{tabular}{llccccccccc}
    \toprule
        \multirow{2}{*}{Model} & \multirow{2}{*}{Training Data} & \multicolumn{4}{c}{NQ320k} & \multicolumn{5}{c}{ExcluIR} \\ \cmidrule(r){3-6} \cmidrule(r){7-11}
        & & R@1 & R@5 & R@10 & MRR & R@1 & MRR & $\mathrm{\Delta}$R@1 & $\mathrm{\Delta}$MRR & RR \\ \hline
        \multirow{2}{*}{DPR} & NQ320k & 54.81  & \textbf{79.50}  & \textbf{85.52}  & 65.39  & 48.55 & 60.50 & 16.45 & 13.49 & 58.76 \\ 
        & N. w/ ExcluIR & \textbf{55.08}  & 79.31  & 85.49  & \textbf{65.58}  & \textbf{55.04} & \textbf{67.89} & \textbf{21.52} & \textbf{16.38} & \textbf{61.00} \\
        \midrule
        \multirow{2}{*}{Sentence-T5} & NQ320k & 59.63  & \textbf{82.78}  & \textbf{87.42}  & \textbf{69.57}  & 57.76 & 66.34 & 32.90 & \textbf{27.96} & 67.83 \\ 
        & N. w/ ExcluIR & \textbf{59.80}  & 81.58  & 87.13  & 69.36  & \textbf{63.09} & \textbf{74.57} & \textbf{34.47} & 26.19 & \textbf{68.00} \\
        \midrule
        \multirow{2}{*}{GTR} & NQ320k & \textbf{62.35}  & \textbf{84.67}  & \textbf{89.17}  & \textbf{71.90}  & 59.79 & 69.00 & 34.85 & 28.12 & 68.31 \\ 
        & N. w/ ExcluIR & 61.44  & 83.82  & 88.34  & 71.01  & \textbf{65.64} & \textbf{76.98} & \textbf{39.05} & \textbf{28.46} & \textbf{69.98} \\
        \midrule
        \multirow{2}{*}{ColBERT} & NQ320k & 60.08  & \textbf{84.19}  & \textbf{89.41}  & \textbf{70.50}  & 57.01 & 70.88 & \textbf{20.02} & \textbf{15.26} & \textbf{59.97} \\ 
        & N. w/ ExcluIR & \textbf{60.20}  & 83.59  & 88.60  & 70.29  & \textbf{57.91} & \textbf{73.52} & 19.30 & 13.05 & 59.71 \\
        \midrule
        \multirow{2}{*}{GENRE} & NQ320k & \textbf{56.25}  & \textbf{71.21}  & \textbf{74.00}  & \textbf{62.80}  & 31.63 & 37.63 & 11.44 & 10.15 & 58.65 \\ 
        & N. w/ ExcluIR & 55.15  & 70.00  & 72.85  & 61.55  & \textbf{65.67} & \textbf{73.01} & \textbf{41.19} & \textbf{20.31} & \textbf{70.48} \\
        \midrule
        \multirow{2}{*}{SEAL} & NQ320k & \textbf{55.24}  & \textbf{75.13}  & \textbf{80.97}  & \textbf{63.86}  & 43.54 & 55.17 & 16.11 & 15.27 & 60.02 \\ 
        & N. w/ ExcluIR & 53.86  & 74.84  & 80.34  & 62.78  & \textbf{70.39} & \textbf{78.40} & \textbf{52.14} & \textbf{43.25} & \textbf{78.02} \\
        \midrule
        \multirow{2}{*}{NCI} & NQ320k & 60.41  & 76.10  & 80.19  & 67.18  & 31.46 & 38.95 & 15.87 & 16.81 & 56.84 \\ 
        & N. w/ ExcluIR & \textbf{60.61}  & \textbf{76.53}  & \textbf{80.55}  & \textbf{67.46}  & \textbf{56.92} & \textbf{64.67} & \textbf{41.13} & \textbf{39.92} & \textbf{72.97} \\
    \bottomrule
    \end{tabular}
    }
\label{aug_results_n}
\vspace{-2mm}
\end{table*}

\section{Limitations of bi-encoder models in ExcluIR with similar positive and negative documents.}
\label{appendix:proof_bi}
Bi-encoder models embed queries and documents into a high-dimensional space to compute the relevance score.
These methods are effective when the semantics of the query and documents are straightforward and do not overlap.
However, in ExcluIR, exclusionary queries contain the semantics of negative documents.
We demonstrate that bi-encoder models struggle to distinguish between positive and negative documents when their vector representations are close in embedding space.
This limitation leads to a bottleneck for bi-encoder models in ExcluIR.
Here is the proof.

\begin{definition}
    $A, B$ can be viewed as queries or documents without loss of generality.

    $q_{A,B}\coloneqq \mathbold{f}(A,B)$ where $\mathbold{f}(A,B)$ is the query encoding vector of query $A$ and negative query $B$.
    $d_{A}\coloneqq \mathbold{g}(A)$ where $\mathbold{g}(\cdot)$ is the document encoding vector function.
    All vectors are normalized to $1$.
    We also define $x,y$ to be $\varepsilon$-close if there exists $\delta\in(0,\frac{1}{9})$ such that $\Pr(\innerprod{x}{y}>1-\varepsilon)>1-\delta$.
\end{definition}

\begin{assumption}
We have certain assumptions.

\begin{itemize}[nosep]
    \item $\da$ and $\db$ are $\varepsilon$-close, which means both $A$ and $B$ are related documents but have some differences that the user would like to distinguish.
    \item $\qab$ and $\da$ are $\varepsilon$-close, which means $\qab$ have good representation to retrieve $\da$.
    Similar is true for $\qba$ and $\db$.
    \item $\innerprod{\qba}{\db}-\innerprod{\qba}{\da}\ge 1-\varepsilon$ w.h.p., which means $\qba$ prefers $\db$ rather than $\da$.
    \item $\varepsilon<3-2\sqrt{2}$.
\end{itemize}
\end{assumption}

\begin{claim}
    We would like to prove that
\begin{itemize}[nosep]
    \item $\innerprod{\qab}{\db}-\innerprod{\qab}{\da}>0$ w.h.p.
\end{itemize}
\end{claim}

\begin{proof}
    \begin{align}
        &\innerprod{\qab}{\db}-\innerprod{\qab}{\da}\\
        &=\innerprod{\qab-\qba}{\db-\da} \\
        &\ \ +\innerprod{\qba}{\db}-\innerprod{\qba}{\da}.
    \end{align}

    Specifically,
    \begin{align}
        &|\innerprod{\qab-\qba}{\db-\da}| \\
        &\le \|\qab-\qba\|\|\db-\da\|\\
        &\le \sqrt{\|\qab\|^2+\|\qba\|^2}\|\db-\da\|\\
        &= \sqrt{2}\sqrt{2-2\innerprod{\db}{\da}}\\
        &\le 2\sqrt{\varepsilon}
    \end{align}

    Therefore, with probability $1-3\delta$ we have
    \begin{align}
        &\innerprod{\qab}{\db}-\innerprod{\qab}{\da}\\
        &\ge -2\sqrt{\varepsilon}+\innerprod{\qba}{\db}-\innerprod{\qba}{\da}\\
        &\ge -2\sqrt{\varepsilon} + 1 - \varepsilon\\
        &>\,0.
    \end{align}
\end{proof}

\section{Additional cases of ColBERT on negative document relevance scoring}
\label{appendix:colbert_case}
We also present more cases for the analysis of ColBERT in Table~\ref{heatmap1}--\ref{heatmap4}.
As we can see, when calculating the relevance score of the negative document, ColBERT always fails to focus on the exclusionary phrases like `aside from' and `other than.'
Instead, it mistakenly focuses on words with lexical matches, resulting in inflated scores.

\begin{figure*}[ht]
    \centering
    \includegraphics[scale=0.32]{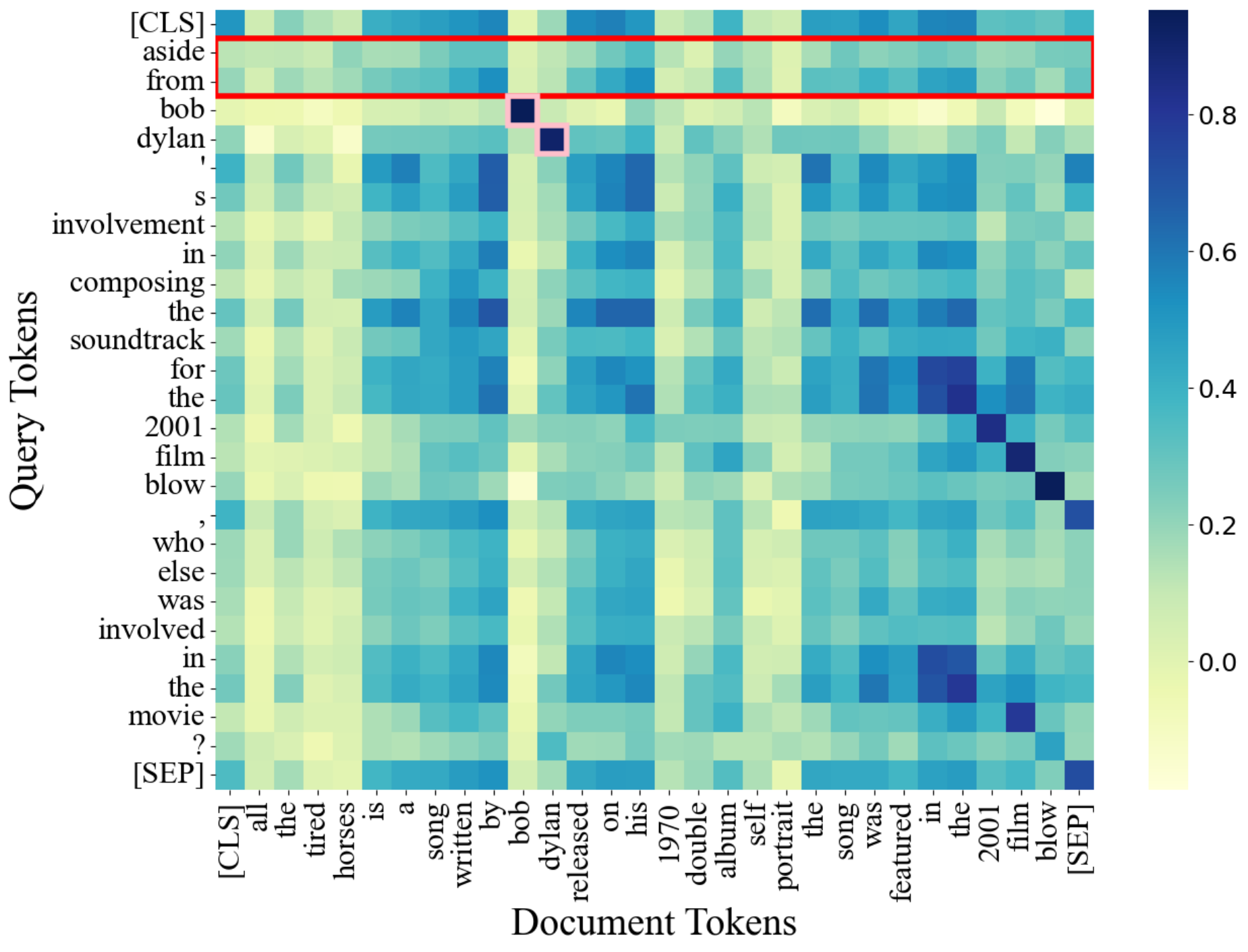}
    \caption{An example of ColBERT on negative document relevance scoring.
    ColBERT overlooks the semantics of `aside from' and instead, due to the presence of lexical matches such as `bob dylan', assigned a high relevance score to this negative document.}
    \label{heatmap1}
\vspace{-2mm}
\end{figure*}
\begin{figure*}[ht]
    \centering
    \includegraphics[scale=0.32]{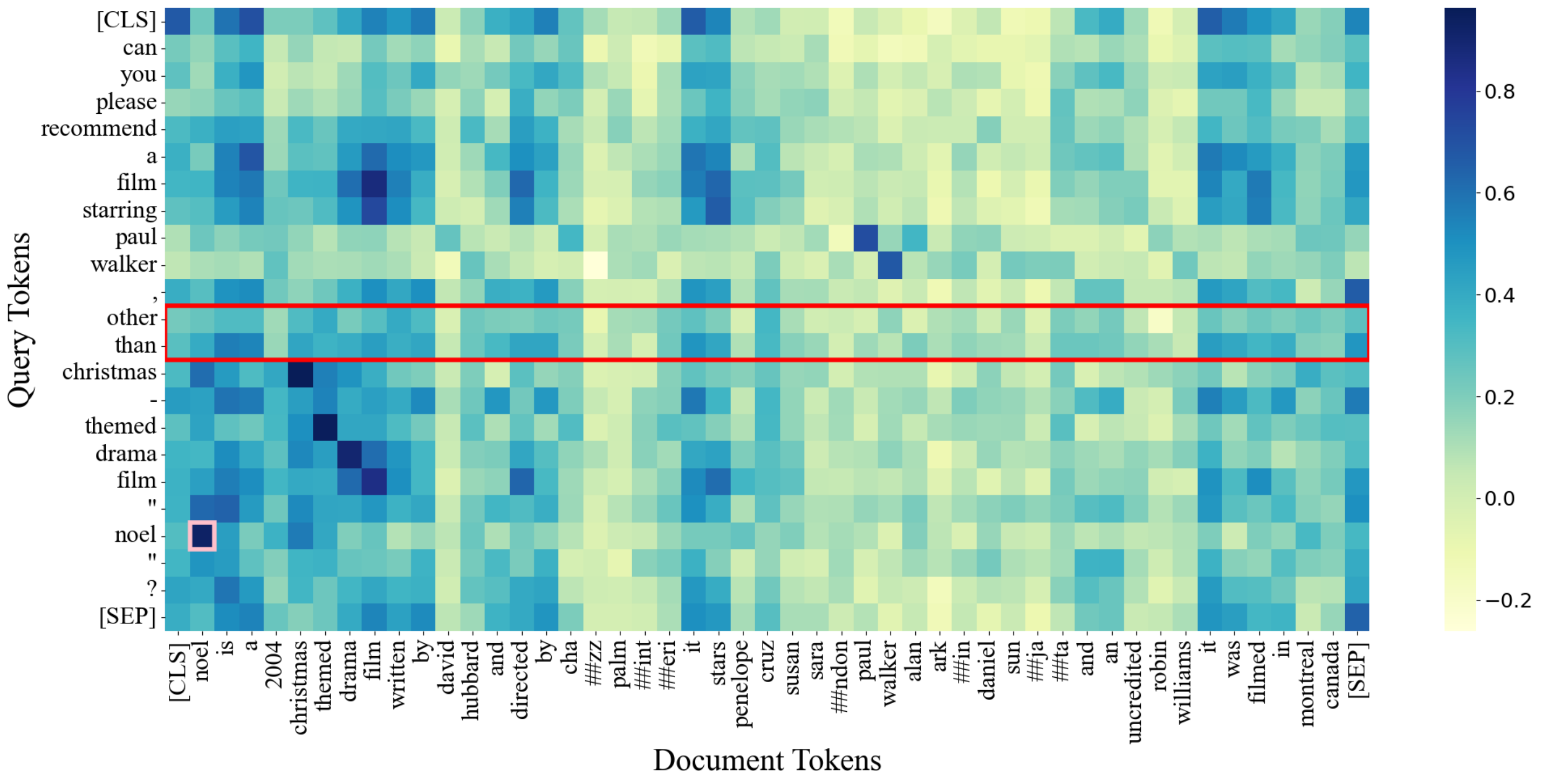}
    \caption{An example of ColBERT on negative document relevance scoring.
    ColBERT overlooks the semantics of `other than' and instead, due to the presence of lexical matches such as `noel', assigned a high relevance score to this negative document.}
    \label{heatmap2}
\vspace{-2mm}
\end{figure*}
\begin{figure*}[ht]
    \centering
    \includegraphics[scale=0.32]{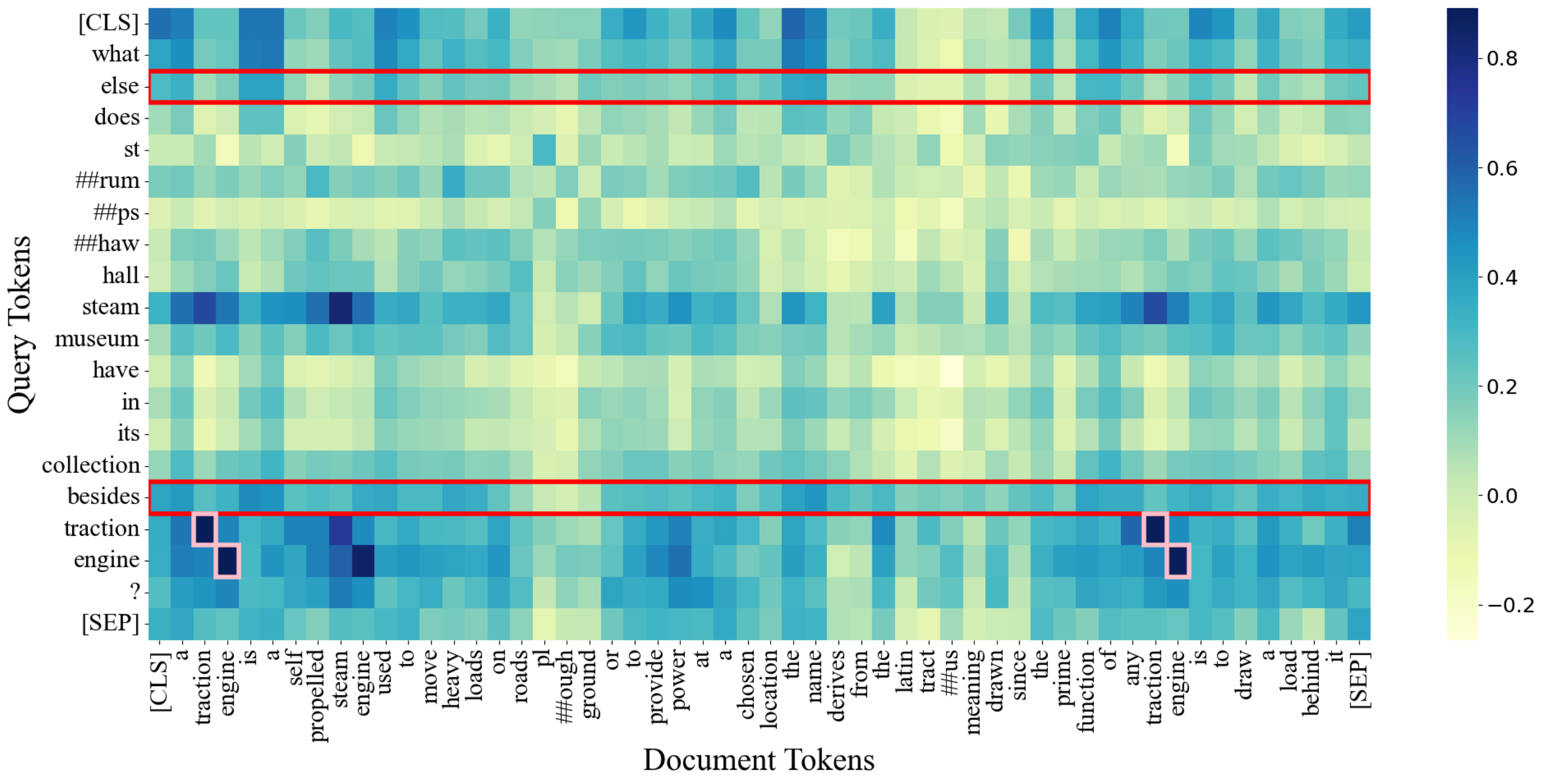}
    \caption{An example of ColBERT on negative document relevance scoring.
    ColBERT overlooks the semantics of `else', `besides' and instead, due to the presence of lexical matches such as `traction engine', assigned a high relevance score to this negative document.}
    \label{heatmap3}
\vspace{-2mm}
\end{figure*}
\begin{figure*}[ht]
    \centering
    \includegraphics[scale=0.32]{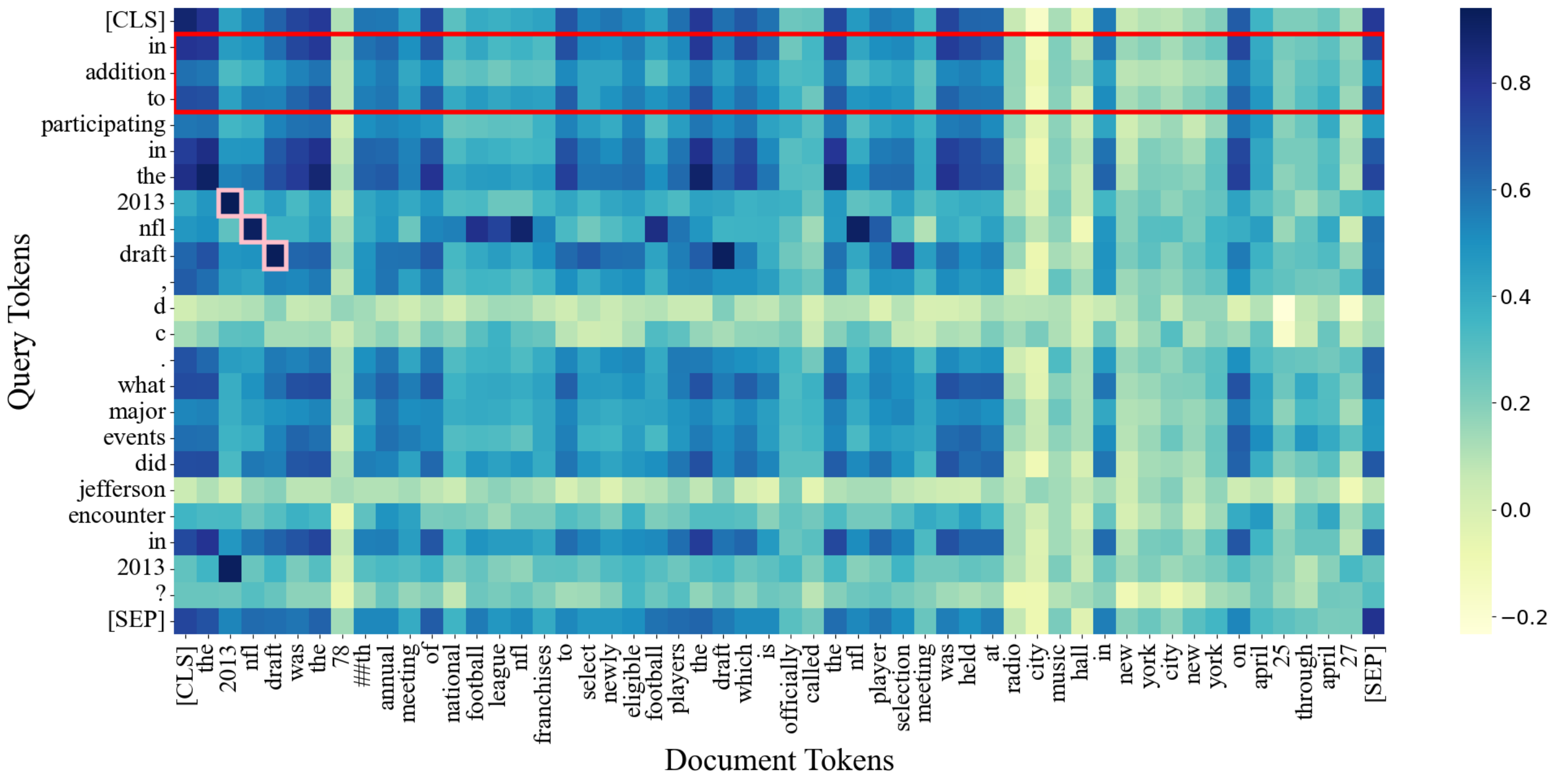}
    \caption{An example of ColBERT on negative document relevance scoring.
    ColBERT overlooks the semantics of `in addition to' and instead, due to the presence of lexical matches such as `2013 nfl draft', assigned a high relevance score to this negative document.}
    \label{heatmap4}
\vspace{-2mm}
\end{figure*}

\section{Cases of ExcluIR dataset}
\label{appendix:cases}

Table~\ref{tab:cases} shows some cases of ExcluIR dataset.

\begin{table*}[ht]
    \centering
    \renewcommand{\arraystretch}{1.2}
    \caption{Cases of ExcluIR dataset.}
    \label{tab:cases}
    \vspace{-3mm}
    \begin{tabular}{p{0.15\linewidth} p{0.7\linewidth}}
    \toprule
    Exclusionary \newline query & Aside from Bob Dylan's involvement in composing the soundtrack for the 2001 film Blow, who else was involved in the movie? \\
    \hline
    Positive \newline document & Blow is a 2001 American biographical crime film about the American cocaine smuggler George Jung, directed by Ted Demme. David McKenna and Nick Cassavetes adapted Bruce Porter's 1993 book ``Blow: How a Small Town Boy Made \$100 Million with the Medellín Cocaine Cartel and Lost It All'' for the screenplay. It is based on the real-life stories of George Jung, Pablo Escobar, Carlos Lehder Rivas (portrayed in the film as Diego Delgado), and the Medellín Cartel. The film's title comes from a slang term for cocaine. \\
    \hline
    Negative \newline document & ``All the Tired Horses'' is a song written by Bob Dylan, released on his 1970 double album ``Self Portrait''. The song was featured in the 2001 film ``Blow''. \\
    \midrule[1pt]
    Exclusionary \newline query & Can you please recommend a film starring Paul Walker, other than Christmas-themed drama film ``Noel''? \\
    \hline
    Positive \newline document & Paul William Walker IV (September 12, 1973 – November 30, 2013) was an American actor. Walker began his career guest-starring in several television shows such as ``The Young and the Restless'' and ``Touched by an Angel''. Walker gained prominence with breakout roles in coming of age and teen films such as ``She's All That'' and ``Varsity Blues'' (1999). In 2001, Walker gained international fame for his portrayal of Brian O'Conner in the street racing action film ``The Fast and the Furious'' (2001), and would reprise the role in five of the next six installments but died in the middle of the filming of ``Furious 7'' (2015). He also starred in films such as ``Joy Ride'' (2001), ``Timeline'' (2003), ``Into the Blue'' (2005), ``Eight Below'', and ``Running Scared'' (2006). \\
    \hline
    Negative \newline document & Noel is a 2004 Christmas-themed drama film written by David Hubbard and directed by Chazz Palminteri. It stars Penélope Cruz, Susan Sarandon, Paul Walker, Alan Arkin, Daniel Sunjata and an uncredited Robin Williams. It was filmed in Montreal, Canada. \\
    \midrule[1pt]
    Exclusionary \newline query & In addition to participating in the 2013 NFL Draft, D C. What major events did Jefferson encounter in 2013? \\
    \hline
    Positive \newline document & D. C. Jefferson (born May 7, 1989) is an American football tight end who is currently a free agent. He played college football at Rutgers University. He was drafted in the seventh round with the 219th overall pick by the Arizona Cardinals in the 2013 NFL Draft. Jefferson was released on November 4, 2013 after he was arrested on suspicion of driving under the influence. \\
    \hline
    Negative \newline document & The 2013 NFL draft was the 78th annual meeting of National Football League (NFL) franchises to select newly eligible football players. The draft, which is officially called the ``NFL Player Selection Meeting,'' was held at Radio City Music Hall in New York City, New York, on April 25 through April 27. \\
    \bottomrule
    \end{tabular}
\vspace{-2mm}
\end{table*}

\end{document}